\documentclass{llncs}
\usepackage[english]{babel}
\usepackage{graphicx}
\usepackage{amssymb, amsmath}

\newtheorem{claimnum}{Claim}
\title{Edge-b-coloring Trees
\thanks{Partially supported by CNPq/Brazil.}
}
\author{Victor Campos \and Ana Silva}
\institute{ParGO Group - Parallellism, Graphs and Optimization\\Universidade Federal do Cear\'a, Fortaleza, Brazil\\\email{campos@lia.ufc.br,\ anasilva@mat.ufc.br}}

\begin{document}

\maketitle

\begin{abstract}
A b-coloring of the vertices of a graph is a proper coloring where each color class contains a vertex which is adjacent to at least one vertex in each other color class. The b-chromatic number of $G$ is the maximum integer $b(G)$ for which $G$ has a b-coloring with $b(G)$ colors. This problem was introduced by Irving and Manlove in 1999, where they showed that computing $b(G)$ is $\mathcal{NP}$-hard in general and polynomial-time solvable for trees. Since then, a number of complexity results were shown, including NP-hardness results for chordal graphs (Havet et. al., 2011) and line graphs (Campos et. al., 2015). In this article, we present a polynomial time algorithm that solves the problem restricted to claw-free block graphs, an important subclass of chordal graphs and line graphs. This is equivalent to solving the edge coloring version of the problem restricted to trees.
\end{abstract}

\section{Introduction}

Let $G$ be a simple graph\footnote{The graph terminology used in this paper follows \cite{BM08}.} 
and let $\psi$ be a proper coloring of
$G$. If a vertex $u$ in a color class $i$ of $\psi$ has no neighbors in another color class $j$, then another proper coloring of $G$ can be obtained by simply changing the color of $u$ to $j$. Therefore, if this holds for every vertex of color class $i$, then we can obtain a proper coloring that uses fewer colors.  Because finding the chromatic number of a graph is $\mathcal{NP}$-hard, one cannot expect to always get a good result by iteratively applying this method.

On the basis of this idea, Irving and Manlove introduced the notion of
b-coloring \cite{Irving.Manlove.99}.  Intuitively, a b-coloring
is a proper coloring that cannot be improved by the described heuristic,
and the b-chromatic number $b(G)$ measures the worst possible such coloring. Finding 
$b(G)$ was proved to be $\mathcal{NP}$-hard in general graphs \cite{Irving.Manlove.99}, 
and remains so even when restricted to bipartite graphs \cite{KRATOCHVIL.etal.02}, chordal
 graphs \cite{HLS.11}, or line graphs \cite{Campos.etal.15}.  In the same article, Irving and Manlove also introduced a simple upper bound for $b(G)$, defined as follows. The \emph{$m$-degree of $G$} is the maximum integer $k$ for which there are at least 
$k$ edges of degree at least $k - 1$; we denote it by $m(G)$. 
Because the existence of a b-coloring with $k$ colors forces the existence of $k$ vertices of degree $k-1$, we get: $$\chi(G)\le b(G) \leq m(G)$$

They then proved that $b(G)\ge m(G)-1$ whenever $G$ is a tree. Since then, it has been discovered that the graph needs only to be ``locally acyclic'' in order to have this property, in other words, it is known that $b(G)\ge m(G)-1$ whenever $G$ has girth at least~7 \cite{Campos.Lima.Silva.15}. A number of other results investigate graphs with high b-chromatic number when compared to $m(G)$; for instance \cite{BMZ.09}, \cite{Cabello.Jakovac.11}, \cite{KM.02}, and \cite{KRATOCHVIL.etal.02}, among others. A natural question is whether this property carries on to the edge version of the problem, that is, whether the b-chromatic number of the line graph $G$ of a tree is also at least $m(G)-1$. The answer to this question is ``yes'', if $G$ is the line graph of a caterpillar \cite{Campos.etal.15}, but it is ``no'' for general trees, as can be seen in \cite{MS.12}. Nevertheless, here we show that deciding whether the b-chromatic number of $G$ is at least $k$ can be done in polynomial time, when $G$ is the line graph of a tree. 
This, together with the result on caterpillars, are the only positive results on a subclass of chordal graphs, up to our knowledge. We believe that our algorithm can be adapted to work on any block graph. However, the generalization to bigger subclasses of chordal graphs, for instance interval graphs, poses much harder difficulties.

Our algorithm was inspired on a previous fixed-parameter algorithm for block graphs presented in~\cite{Silva.10} (we mention that the fixed-parameter decision problem is open for general graphs).
Consider $W\subseteq D_k(G)$ with cardinality $k$, and, for each $u\in W$, consider a subset $N_u$ of size $k-1$. Let $G'$ be the graph obtained from $G$ by turning each subset $N_u$ into a clique; also, let $\psi$ be the precoloring of $G'$ where each $u\in W$ is colored with a distinct color, and no other vertex is colored. Note that if $\psi$ can be extended to the whole graph $G'$, then the same extension is a b-coloring of $G$ having $W$ as basis. In~\cite{Marx.07}, Marx proves that if $G'$ is a chordal graph, then deciding whether the precoloring can be extended to $G'$ can be done in polynomial time. Sampaio and Silva observed that if $G$ is a block graph, then the obtained graph $G'$ is chordal. They presented their result separatedly in~\cite{HLS.11} and \cite{Silva.10}, and in \cite{Silva.10} it is observed that this leads to a fixed-parameter-algorithm for block graphs: it suffices to test whether the extension exists, for every possible subset $W$, and for every family of subsets $\{N_u\mid u\in W\}$.

Here, we consider line graphs of trees (which equals claw-free block graphs) and use a flow network representation of a b-coloring of $G$ (this is similar to the one used by Marx in~\cite{Marx.07}); then, we use a dynamic programming algorithm to eliminate the need for an exponential number of tests. We believe that our algorithm can be made to work on general block graphs, and maybe even on larger classes, provided that the auxiliary graph constructed as in the previous paragraph is still chordal. However, we do not believe that this class is much larger; for instance, if $u\in W$ has a $P_4$, ($u_1,u_2,u_3,u_4$), in its neighborhood and $N_u\cap\{u_1,u_2,u_3,u_4\} = \{u_1,u_4\}$, then already we get a $C_4$ in the auxiliary graph. 

Before we start, we need some formal definitions. Let $G$ be a graph. A $k$-coloring of $G$ is a function $\psi:V(G)\rightarrow\{1,\cdots,k\}$, and it is said to be \emph{proper} if $\psi(u)\neq \psi(v)$, whenever $uv\in E(G)$. Consider $\psi$ to be a proper $k$-coloring of $G$. The \emph{color class $i$ (of $\psi$)} is the subset of vertices $\{u\in V(G)\mid \psi(u) = i\}$. We say that vertex $u\in V(G)$ \emph{realizes color $i$} (or that \emph{color $i$ is realized by $u$}) if $\psi(u)=i$ and $u$ is adjacent to at least one vertex of color class $j$, for every $j\in \{1,\cdots,k\}\setminus\{i\}$. We say that a subset \emph{$W$ realizes distinct colors} if each vertex of $W$ realizes a color, and the set of realized colors equals $\lvert W\rvert$. If $\psi$ is such that every color is realized by some vertex, then we say that $\psi$ is a \emph{b-coloring of $G$}. Also, a subset containing exactly one vertex that realizes $i$, for each color class $i$, is called a \emph{basis of $\psi$}. A vertex $v\in V(G)$ is said to be \emph{$k$-dense} if $d(v)\ge k-1$, and the set of all $k$-dense vertices of $G$ is denoted by $D_k(G)$. Clearly, if $\psi$ is a b-coloring of $G$ with $k$ colors with basis $W$, then $W\subseteq D_k(G)$. The \emph{b-chromatic number of $G$}, denoted by $b(G)$, is the maximum integer $k$ for which $G$ has a b-coloring with $k$ colors. 

A graph is called a \emph{block graph}  if its 2-connected components are cliques, and it is \emph{claw-free} if it does not have an induced subgraph isomorphic to $K_{1,3}$ (the complete bipartite graph with parts of size 1 and 3, also known as ``claw''). It is well known, and not hard to verify, that the class of claw-free block graphs equals the class of line graphs of trees. In the remainder of the text, we refer only to claw-free block graphs, and consider the graph to be rooted at some block. 

Here, we want to decide, given a claw-free block graph $G$ and an integer $k$, whether $b(G)\ge k$. Let $\omega(G)$ denote the size of a largest clique in $G$. Clearly, the answer is always ``no'' when $k<\omega(G)$, and, because $G$ is a chordal graph (which means that $\omega(G)=\chi(G)\le b(G)$), the answer is ``yes'' when $k= \omega(G)$. Thus, from now on we suppose that $k>\omega(G)$. 
To solve the problem, we compute the maximum size of a subset $W$ that realizes distinct colors on a $k$-coloring of $G$. Clearly, we will get $b(G)\ge k$ if and only if the maximum size of such a set is $k$. The idea is then to compute this value for smaller subgraphs of $G$ and then try to combine these solutions. For this, we need a more convenient way to represent a solution. In the next section, given a subset $W$ that can realize distinct colors in a $k$-coloring, we show another representation of such a $k$-coloring related to $W$. Then, in the following section, we show how to combine the partial solutions. 

\section{Representation of a $k$-coloring that realizes $\lvert W\rvert$ colors}

Before we start, we need some further definitions. Consider a claw-free block graph $G$ rooted at some block $B$, an integer $k>\omega(G)$, and a subset $W\subseteq D_k(G)$. For any block $B'$ different from the root, denote by $P(B')$ the \emph{parent block of $B'$}, and by $c(B')$ the cut vertex that separates $B'$ from $P(B')$. Finally, denote by $G_{B'}$ the subgraph rooted at $B'$. 

The \emph{flow network related to $(G,W)$} is denoted by $\mathcal{F}(G,W)$ and is obtained as follows (denote by $c$ the capacity function). Consider a block $B'=\{x_1,\cdots,x_q\}$ of $G$. Add nodes $\{B'_{x_1},\cdots,B'_{x_q},(B')\}$ to $\mathcal{F}(G,W)$, and set the capacity of each node to one, except for $(B')$ which has capacity $k-q$. Denote this set of nodes by $V(B')$ and call node $(B')$ the \emph{cash-node (of $B'$)}. Suppose, without loss of generality, that $x_1,\cdots,x_p$, $p\le q$, are all the cut vertices in $B'$ different from $c(B')$. For each $i\in\{1,\cdots,p\}$, let $B^i$ be the block containing $x_i$ different from $(B')$ (i.e., $P(B^i)=B'$ and $x_i = c(B^i)$). 
\begin{itemize}
 \item[(I)] If $x_i\not\in W$ (observe Figure \ref{xnotinW}) - add node $(B^i,B')$ with capacity $k-1$, and the following arcs:
 \begin{enumerate}
  \item $(B^i_x,(B^i,B'))$, for all $x\in B^i\setminus\{x_i\}$; 
  \item $((B^i,B'),B'_x)$, for all $x\in B'\setminus\{x_i\}$; 
  \item $(B^i_{x_i},B'_{x_i})$, $((B^i),(B^i,B'))$, and $((B^i,B'),(B'))$.
 \end{enumerate}

 \begin{figure}\label{xnotinW}
\begin{center}
\includegraphics[height=3cm]{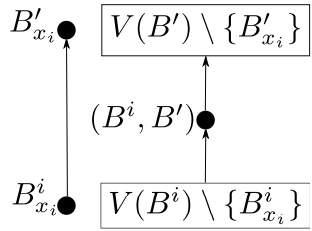} 
\end{center}
\caption{Flow network representation when $x_i\notin W$. The arcs represent the existence of every arc between the corresponding subsets of nodes.}
\end{figure}

 \item[(II)] If $x_i\in W$ (observe Figure \ref{xinW}) - add node $(B^i,B')$ with capacity $d_G(x_i)-k+1$, a source node $\langle x_i\rangle$ that sends out 1 unit of flow, and the following arcs:
 \begin{enumerate}
  \item $(B^i_x,(B^i,B'))$ and $(B^i_x,(B'))$, for all $x\in B^i\setminus \{x_i\}$;
  \item $((B^i),B'_x)$ and $((B^i,B'),B'_x)$, for all $x\in B'\setminus \{x_i\}$; and
  \item $(B^i_{x_i},B'_{x_i})$ and $(\langle x_i\rangle, B^i_{x_i})$.
 \end{enumerate}
 
 \begin{figure}\label{xinW}
\begin{center}
\includegraphics[height=3.5cm]{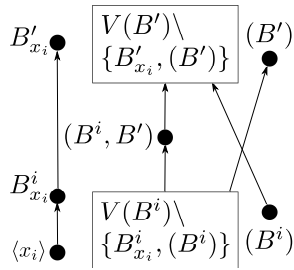}
\end{center}
\caption{Flow network representation when $x_i\in W$. The arcs represent the existence of every arc between the corresponding subsets of nodes.}
\end{figure}

\item[] Nodes of type $(B^i,B')$ are also called \emph{cash-nodes}.
\end{itemize}

Finally, we set $V(B)$ to be the sinks. The network will admit a flow if each flow-unit coming from a source can be rooted to the sinks. More formally, write $W=\{w_1,\cdots,w_{\kappa}\}$, $\kappa\le k$, and let ${\cal N}$ denote the set of non-source nodes of $\mathcal{F}(G,W)$. A \emph{flow in $\mathcal{F}(G,W)$} is a collection of directed paths $\{p_1,\cdots,p_\kappa\}$ of $\mathcal{F}(G,W)$ such that: $p_i$ starts at $\langle w_i\rangle$ and ends at $V(B)$, for every $i\in \{1,\cdots,\kappa\}$; and $\lvert \{i\mid \gamma\in p_i\}\rvert\le c(\gamma)$, for every $\gamma\in {\cal N}$. We denote the value $\lvert \{i\mid \gamma\in p_i\}\rvert$ by $f(\gamma)$.
We want to prove that there exists a $k$-coloring where $W$ realizes distinct colors if and only if $\mathcal{F}(G,W)$ admits a flow.

Intuitively, one can see the flow paths as colors that are forced to vertices because of the precoloring. When no implications exist, we can ``store flow'' in the cash-nodes, which means that no vertex in that block must necessarily have that color. That is why the flow path containing node $B^i_{x_i}$ above must also contain node $B'_{x_i}$, i.e., the color of a node cannot change from one block to the other. Also, if $x_i\in W$, then the colors that do not necessarily appear in block $B^i$ must appear in block $B'$ in order for $x_i$ to realize its color. Note that some of the colors that appear in $B^i$ are allowed to appear in $B'$ but up to a limit set by the capacity of node $(B^i,B')$, which depends on the degree of $x_i$. This is because if too many colors are repeated in the neighborhood of $x_i$ then it may be impossible for $x_i$ to realize its color. In what follows, we give a formal definition of these ideas.

Let $f = \{p_1,\cdots,p_\kappa\}$ be a flow in $\mathcal{F}(G,W)$. For each vertex $x\in V(G)$, let ${\cal N}(x)$ represent the set of nodes of $\mathcal{F}(G,W)$ related to $x$, i.e., $B_x\in {\cal N}(x)$, for every block $B$ containing $x$, and $\langle x\rangle\in {\cal N}(x)$ if $x\in W$. Let $\psi$ be a $k$-coloring of $G$. We say that $x\in V(G)$ is \emph{saturated by $p_i$ (in $f$)} if ${\cal N}(x)\cap p_i\neq \emptyset$, and that $x$ is \emph{flow colored} (in $\psi$ according to $f$) if $x$ is saturated by $p_{\psi(x)}$. If $\psi$ is a proper $k$-coloring such that every saturated vertex is flow colored, then we say that $\psi$ is a \emph{flow coloring of $f$}.



\begin{theorem}\label{thm:flow}
Let $W\subseteq D_k(G)$, and $F=\mathcal{F}(G,W)$. If $F$ has a flow $f$, then there exists a flow coloring $\psi$ of $G$ where each $w_i\in W$ realizes a distinct color. Conversely, if $\psi$ is a coloring of $G$ such that each $w_i\in W$ realizes a distinct color, then $\psi$ is a flow coloring of $f$, for some flow $f$ in $F$.
\end{theorem}
\begin{proof}
The proof is by induction on the height of $G$. If $G = B$, because $k>\omega(G)$, we get $D_k(G) = \emptyset$ and the theorem is vacuously true. So, let $X=\{x_1,\cdots,x_p\}$ be the cut vertices of $B$, and, for every $i\in \{1,\cdots,p\}$, let $B^i$ be the block containing $x_i$ other than $B$, $G_i$ be subgraph of $G$ rooted at $B^i$, $W_i$ be the set $W\cap V(G_i)$, and $F_i$ be the subnetwork of $F$ related to $G_i$. 

  First, consider a flow $f$ in $F$ and, for each $i\in \{1,\cdots, p\}$, let $f_i$ be the flow $f$ restricted to $F_i$. By induction hypothesis, there exists a flow coloring $\psi_i$ of $G_i$ such that each $w\in W_i$ realizes a distinct color. We can suppose that the colors realized in $\psi_i,\psi_j$ are distinct whenever $i\neq j$, as otherwise it suffices to rename some colors. Let $\psi$ be the coloring (not necessarily proper) obtained by the union $\bigcup_{i=1}^p\psi_i$ (it is well defined since $V(G_i)\cap V(G_j)=\emptyset$ whenever $i\neq j$). In order to obtain a flow coloring, we need to ensure that: 1) each saturated $x\in B$ is flow colored; 2) the coloring is proper; and 3) that each $x\in W\cap B$ realizes a distinct color.
 
First, consider a saturated vertex $x\in B$. By the capacity of nodes in ${\cal N}(x)$ and the construction of $F$, one can verify that $x$ is saturated by at most one path, say $p_\ell$. We want to ensure that $x$ is colored with $\ell$. If $x$ is not a cut vertex, then it is not colored in $\psi$ and we can just color $x$ with $\ell$; so, suppose that $x=x_i$, for some $x_i\in X$. If $x_i\in W$, note that the path $(\langle x_i\rangle, B^i_{x_i})$ must be in $f_i$ and, by induction hypothesis, $x_i$ is flow colored.
So, consider that $x_i\notin W$, and suppose that $\psi(x_i)=j\neq \ell$. By induction hypothesis, we know that $x_i$ is not saturated in $f_i$, i.e., $B^i_{x_i}$ is not contained in $p_\ell$. Also, because $(B^i,B)$ is not adjacent to $B_{x_i}$, we know that $p_\ell$ also does not contain $(B^i,B)$, which implies that $w_\ell\notin W_i$. The following claim ensures us that we can change $\psi_i$ to color $x_i$ with $\ell$.

\begin{claimnum}\label{claim1}
 Let $x_i$ be nonsaturated in $f_i$ with $\psi_i(x_i)=j$, and let $\ell\in\{1,\cdots,k\}\setminus \{j\}$ such that $w_\ell\notin W_i$. Then, we can change the color of $x_i$ from $j$ to $\ell$ in $\psi_i$.
\end{claimnum}
\begin{proof}
Note that if $w_j\notin W_i$, we can simply switch colors $j$ and $\ell$ in $G_i$; so suppose otherwise. Because $\psi_i$ is a flow coloring of $G_i$ with respect to $f_i$, which means that it is proper and that every saturated vertex is flow colored, we know that path $p_j$ contains cash-node $(B^i)$; hence, $w_j\notin B^i$. Let $x$ be the cut vertex that separates $x_i$ from $w_j$, and recall that $\psi_i(x)\neq j$, since $\psi_i$ is proper. Denote by $G_x$ the graph rooted at $x$ and by $G'_x$ the graph $G_i-G_x$. If $x\notin W$ and $\psi(x)\neq\ell$, then we can switch colors $j$ and $\ell$ in $G'_x$. So, we analyse the following cases:

\begin{enumerate}
  \item $x\in W$: Let $B'$ be the block containing $x$ different from $B^i$. Because $(B')$ is not adjacent to $(B^i)$, we get that path $p_j$ must contain $B'_y$, for some $y\in B'\setminus\{x\}$, which means that color $j$ is repeated in $N(x)$. Thus, we can just switch colors $j$ and $\ell$ in $G'_x$;
	
  \item $\psi(x)=\ell$: Note that this implies $x\notin W$ since $w_\ell\notin W_i$. Let $W_x=W\cap V(G_x)$. If there exists some color $\ell'$ such that $\ell'\notin \psi(B^i)$ and $w_{\ell'}\notin W_x$, then we switch colors $\ell$ and $\ell'$ in $G_x$ and colors $j$ and $\ell$ in $G'_x$. Thus, suppose otherwise. This means that every color appears in $\psi(B^i)$ or in $\psi(W_x)$, i.e., \[\lvert \psi(B^i \cup W_x)\rvert = k\Rightarrow \lvert W_x\rvert = k-\lvert B^i\rvert +\lvert \psi(B^i)\cap \psi(W_x)\rvert\]
  Because the capacity of $(B^i)$ is $k-\lvert B^i\rvert$, we get that at least $\lvert \psi(B^i)\cap \psi(W_x)\rvert$ paths starting at $W_x$ must pass through nodes in $V(B^i)\setminus (B^i)$ . Also, because $\psi_i$ is a flow coloring, at most $\lvert \psi(B^i)\cap \psi(W_x)\rvert$ nodes in $V(B^i)\setminus (B^i)$ are contained in paths starting at $W_x$. This means that every node in $V(B^i)\setminus(B^i)$ that receives color from $\psi(W_x)$ is flow colored, a contradiction since $\psi(x_i) = j \in \psi(W_x)$ and $x_i$ is not flow colored.
\end{enumerate}

\end{proof}

Now, we want to ensure that $\psi$ is a proper coloring. For this, consider $x_i$ to be an unsaturated vertex colored with some color that also appears in $\psi(B\setminus\{x_i\})$. Note that, if there is any conflict, such a vertex must exist, since every saturated vertex is flow colored. We show that  there exists a color $\ell\in \{\kappa+1,\cdots,k\}\setminus \psi(B)$, and apply Claim \ref{claim1} to change the color of $x_i$ to $\ell$ in $\psi_i$, thus eliminating the conflict. So suppose otherwise, i.e., that $\{\kappa+1,\cdots,k\}\subseteq \psi(B)$. Also, let $C$ denote the set $\{j\in\{1,\cdots,\kappa\}\mid (B)\in p_j\}$, and observe that $\{1,\cdots,\kappa\}\setminus \psi(B)\subseteq C$. Thus, every color appears in $\psi(B)$ or in $C$, i.e., $\lvert \psi(B)\cup C\rvert \ge k$. However, by the capacity of $(B)$ and the fact that $x$ has a conflicting color, we get that $\lvert \psi(B)\cup C\rvert \le \lvert \psi(B)\rvert + \lvert C\rvert\le \lvert B\rvert -1 + k -\lvert B\rvert$, a contradiction. Note that a similar argument can be applied if $x$ has a color of $C$, or if $x$ is not colored. Because of this, we can suppose that every unsaturated vertex of $B$ is colored with a color not in $\{1,\cdots,\kappa\}$



Finally, we need to prove that we can realize a disctinct color in each vertex of $W\cap B$. Because $k>\omega(G)$ and $W\subseteq D_k(G)$, we know that $W\cap B\subseteq X$. Consider $x_i\in W\cap B$, and denote by $M_i$ the set of colors that do not appear in $N(x_i)$, and by $C_i$ the set of colors $\psi_i(W_i)$. We change $\psi_i$ in order to decrease the number of colors in $M_i$. First, note that, for every $j\in C_i$, either $p_j$ intersects $V(B^i)\setminus \{(B^i)\}$ or it contains $(B^i)$ and, hence, must intersect $V(B)\setminus \{(B)\}$. Therefore, since we already know that every vertex is flow colored, we get that $M_i\cap C_i=\emptyset$. Now, consider any $\ell\in M_i$, and let $D_i$ be the set of colors of $C_i$ that are repeated in $N(x_i)$, i.e., $D_i = \{j\in C_i\mid j\in \psi(B^i)\cap \psi(B)\}$. By the capacity of node $(B^i,B)$, we know that $d(x_i)\ge k-1+\lvert D_i\rvert$. Therefore, some color $\ell'$ not in $C_i$ must be repeated in the neighborhood of $x_i$. Since $\ell$ is also not in $C_i$, we can just switch colors $\ell$ and $\ell'$ in $\psi_i$.

Now, we prove the second part of the theorem. Let $\psi$ be a coloring of $G$ where each $w_i\in W$ realizes a distinct color and let $\psi_i$ be the coloring $\psi$ restricted to $G_i$, for each $x_i\in X$. By induction hypothesis, $\psi_i$ is the flow coloring of a flow in $F_i$; because $F_i$, $i\in \{1,\cdots,p\}$, are pairwise disjoint, we make an abuse of language and denote by $p_j$ the path starting in $\langle w_j\rangle$ and ending in $\bigcup_{i=1}^p V(B^i)$. We want to extend these paths into a flow in $F$ for which $\psi$ is a flow coloring. First, we just extend them without paying atention to the nodes' capacities. Consider any $j\in\{1,\cdots,\kappa\}$ and let $\gamma$ be equal to $B_x$, if there exists $x\in B$ such that $\psi(x)=j$; otherwise, let $\gamma$ be equal to $(B)$. We want to ensure that path $p_j$ contains $\gamma$. Let $x_i\in X$ be such that $w_j\in W_i$. If $\gamma=B_{x_i}$, by induction hypothesis, we know that $p_j$ contains $B^i_{x_i}$ and we just add $\gamma$ to $p_j$; hence, suppose $\gamma\neq B_{x_i}$. Because $\psi$ is a proper coloring and by induction hypothesis, we know that $\psi(x_i)\neq j$, and: either (1) $p_j$ contains $B^i_{x'}$, for some $x'\in B^i\setminus\{x_i\}$; or (2) $p_j$ contains $(B^i)$. If (1) occurs and either $\gamma\neq (B)$ or $x_i\notin W$, or if (2) occurs and $x_i\notin W$, then add $(B^i,B)$ and $\gamma$ to $p_j$. If (1) occurs, $\gamma=(B)$, and $x_i\in W$, then add $\gamma$ to $p_j$. Finally, if (2) occurs and $x_i\in W$, because $x_i$ realizes its color and $j\notin \psi(B^i)$, we must have that $\gamma\neq (B)$; hence, we can just add $\gamma$ to $p_j$.

It remains to prove that the capacities are respected. Because $W_i\cap W_j=\emptyset$, whenever $i\neq j$, we get that $f(B_x)\le 1$, for all $x\in B$. So, it remains to prove that the capacities of the following cash-nodes are respected:
\begin{itemize}
 \item cash-node $(B^i,B)$, when $x_i\notin W$: its capacity is $k-1$ and it is violated only if all $\kappa=k$ and no flow path passes through $B^i_{x_i}$. However, $x_i$ is colored with some color, which contradicts the fact that $\psi_i$ is a flow coloring;
 
 \item cash-node $(B^i,B)$, when $x_i\in W$: in this case, its capacity is $d(x_i)-k+1$. Because $(B^i,B)$ is not adjacent to $(B)$ nor to $(B^i)$, we know that each path passing through $(B^i,B)$ must come from a node $B^i_{x'}$ and end at a node $B_x$. Since $x,x'$ are flow colored, this means that each path passing through $(B^i,B)$ defines a color that is being repeated in $N(x_i)$. Since $x_i$ realizes a color in $\psi$, we know that at most $d(x_i)-k+1$ colors are repeated in its neighborhood, i.e., that $f(B^i,B)\le c(B^i,B)$;
 
 \item cash-node $(B)$: each flow path ending at $(B)$ defines a color in $\{1,\cdots,\kappa\}$ that is not used in $B$. Because $B$ is a clique, there are at most $k-\lvert B\rvert$ such colors.
\end{itemize}
\end{proof}

\section{Finding the maximum $\lvert W\vert$}\label{sec:dynalg}

Now, we want to find tha maximum size of a subset $W\subseteq D_k(G)$ that realizes distinct colors. We solve the problem for each subgraph rooted at some block, starting by the leaf nodes and going up towards the root $B$. Recall that $k>\omega(G)$, which implies that the answer is always zero on the leaf nodes. For simplicity, suppose we are at the root $B=(x_0,\cdots,x_q)$, where $x_0$ connects $B$ to its parent block, if it exists, and $x_1,\cdots,x_p$ are all the cut vertices in $B$, for some $p\in\{1,\cdots,q\}$. We denote by $\mathcal{F}(B)$ the family of flow networks $\{\mathcal{F}(G,W)\mid W\subseteq D_k(G)\}$. We make an abuse of language and say that there exists a flow $f$ in $\mathcal{F}(B)$ if $f$ is a flow in $\mathcal{F}(G,W)$, for some $W\subseteq D_k(G)$, and the \emph{value of $f$} is given by $\lvert W\rvert$. Also, for each $i\in\{1,\cdots,p\}$, we denote by $B^i$ the block containing $x_i$ other than $B$. Now, let $b\in\{0,1\}$ and $j\in \{0,\cdots,k-q\}$. We say that a flow $f\in \mathcal{F}(B)$ \emph{realizes $(b,j)$} if $f(B_{x_0})=b$, and $f((B))=j$, and we define:

\begin{itemize}
\item[*] $S_B(b,j)$: maximum value $i$ for which there exists a flow $f$ in $\mathcal{F}(B)$ of value $b+j+i$ that realizes $(b,j)$.
\end{itemize}

We want to compute table $S_B$ using tables $S_{B^1},\cdots,S_{B^p}$. Note that no vertex in $B$ can define a source in these subsolutions since they have degree at most $\omega(G)<k$ in the related subgraphs. Therefore, we need to investigate the possibility of adding a subset of $\{\langle x_1\rangle, \cdots,\langle x_p\rangle\}$ as new sources. However, there is an exponential number of subsets to investigate. Because of this, we need some auxiliary tables. For each $i\in \{1,\cdots,p\}$, let $H_i$ be the graph which corresponds to the component of $G - \bigcup_{j=i+1}^p E(B^j)$ containing $B$, and let $\mathcal{F}_i$ be the family of flow networks $\{\mathcal{F}(H_i,W)\mid W\subseteq D_k(H_i)\}$. 
Now, consider $\ell'\in\{1,\cdots,p\}$, $b\in\{0,1\}$, $j_1\in\{0,\cdots,\ell'\}$, and $j_2\in\{0,\cdots,k-q\}$.  We say that a flow $f\in {\cal F}_{\ell'}$ \emph{realizes $(b,j_1,j_2)$} if $f(B_{x_0})=b$, $f(\{B_{x_1},\cdots,B_{x_{\ell'}}\})=j_1$, and $f((B)) = j_2$. Then, we define:

\begin{itemize}
\item[*] $P_{\ell'}(b,j_1,j_2)$: maximum value $i$ for which there is a flow in $\mathcal{F}_{\ell'}$ of value $b+j_1+j_2+i$ that realizes $(b,j_1,j_2)$.
\end{itemize}
 
Because $H_p$ equals $G$, we get:
\[S_B(b,j)=\max_{0\le j_1\le p}\{P_p(b,j_1,j)+j_1\}\]

In what follows, sometimes we implicitly assume that a flow strictely contained in some other exists. That holds because of the following proposition (it suffices to ignore one of the paths).

\begin{proposition}\label{prop:1lessflow}
If there exists a flow $f$ in $\mathcal{F}(H,W)$, where $H$ is rooted at $R$, then there exists a flow $f'$ in $\mathcal{F}(H,W\setminus\{w\})$, $\forall w\in W$.
\end{proposition}

Now, considering that we know tables $P_{\ell'}$ and $S_{B^{\ell'+1}}$, we want to compute $P_{\ell'+1}(b,j_1,j_2)$, where $b\in\{0,1\}$, $j_1\in\{0,\cdots,\ell'+1\}$, and $j_2\in\{0,\cdots,k-q\}$. For this, we need to analyse what types of solutions in $S_{B^{\ell'+1}}$ and in $P_{\ell'}$ can be combined. Consider $f'\in F_{\ell'}$ realizing entry $e'=(b',j'_1,j'_2)$ of $P_{\ell'}$, where $b'\le b$, and $f''\in \mathcal{F}(B^{\ell'+1})$ realizing entry $e''=(b'',j'')$ of $S_{B^{\ell'+1}}$. Let the values of $f',f''$ be $b'+j'_1+j'_2+w'$ and $b''+j''+w''$, respectively. We want to construct a flow $f\in \mathcal{F}_{\ell'+1}$ that realizes $(b,j_1,j_2)$. 
Let $W'$ and $W''$ be such that $f'$ is a flow in ${\cal F}(H_{\ell'},W')$ and $f''$ is a flow in ${\cal F}(G_{\ell'+1},W'')$, where $G_{\ell'+1}$ is the subgraph rooted at $B^{\ell'+1}$. Intuitively, what we do is applying Proposition \ref{prop:1lessflow} to accomodate $f'\cup f''$ into a flow in ${\cal F}(H_{\ell'+1}, W)$, for some $W\subseteq W'\cup W''$. We also try to increase the combined flow's value by adding $x_{\ell'+1}$ to $W$. Below, by ``push flow into node $\gamma$'', we mean that we increase the corresponding path to a path with extremity in $\gamma$. Observe that every path in $f'$ already has an extremity in $V(B)$; therefore, we try to push flow $f''$ into $V(B)$ taking into consideration $f'$ and the values $b,j_1,j_2$.

First, we try to push the flow $f''(V(B^{\ell'+1}))$ into $V(B)$ without trying to add vertex $x_{\ell'+1}$ as a new source. In this case, all the flow not in $B^{\ell'+1}_{x_{\ell'+1}}$ can be pushed into whatever unsaturated node in $V(B)\setminus\{B_{x_{\ell'+1}}\}$ by passing through $(B^{\ell'+1},B)$. So, if either $j'_1+b''>j_1$ or $j'_2>j_2$, we know that the produced flow will not realize $(b,j_1,j_2)$. This is also the case when one of the following situations occurs. If the amount of flow in $V(B^{\ell'+1})\setminus B^{\ell'+1}_{x_{\ell'+1}}$ is not sufficient to satisfy the following lacks: of $(b-b')$ in $B_{x_0}$; of $j_1-j'_1-b''$ in $\{B_{x_1},\cdots,B_{x_{\ell'}}\}$; and of $j_2-j'_2$ in $(B)$. Or if $j_1=\ell'+1$, $b''=0$, and the amount of flow $w'$ in $V(B)\setminus\{(B),B_{x_1},\cdots,B_{x_{\ell'}}\}$ is not sufficient to satisfy the lack in $B_{x_{\ell'+1}}$. We then say that $f',f''$ are \emph{weakly $(b,j_1,j_2)$-compatible} if 

\begin{itemize}
  \item[W1] $d_1 = j_1 - j'_1 -b'' \ge 0$; 
  \item[W2] $d_2 = j_2 - j'_2\ge 0$; 
  \item[W3] $w''+j''\ge d_1+d_2+(b-b')$, and $w'\ge d'_1 = \min\{0,j_1 - (\ell'+b'')\}$. 
\end{itemize}


If $f'$ and $f''$ are weakly $(b,j_1,j_2)$-compatible, their \emph{combined value}, denoted by $v(f',f'')$, is the amount of flow that can be sent to $V(B)\setminus\{B_{x_0},\cdots,B_{x_{\ell'+1}},(B)\}$. By applying Proposition \ref{prop:1lessflow}, this either equals the number of nodes, or the sum $w'+w''+j''$ minus the quantity of flow sent to satisfy the lacks. That is: $v(f',f') = \min\{q-(\ell'+1), w'+w''+j'' - ((b-b') + d_1 +  d'_1 + d_2)\}$. 

Now, if we want to turn $x_{\ell'+1}$ into a source, we have to pick entries in $S_{B^{\ell'+1}}$ of type $(0,j'')$ and push the flow $j''$ into $V(B)\setminus(B)$, and the flow $S_{B^{\ell'+1}}(0,j'')$ into $(B)$. Suppose we can push $r_1$ units of flow through $(B^{\ell'+1},B)$. By similar arguments, we need: the amount of flow lacking in $\{B_{x_0},\cdots,B_{x_{\ell'}}\}$ to be nonnegative and to be satisfiable by the amount of flow in $(B^{\ell'+1})$ plus $r_1$; and the amount of flow lacking in $(B)$ to be non-negative and to be satisfiable by the amount of flow in $V(B^{\ell'+1})\setminus\{(B^{\ell'+1})\}$ minus $r_1$. For this, we define $f',f''$ to be \emph{strongly $(b,j_1,j_2)$-compatible} if $b''=0$ and there exists $r_1$ such that:

\begin{itemize}
  \item[S1] $0\le r_1\le d_G(x_{\ell'+1}) - k + 1$;
  \item[S2] $0\le d_1 = j_1 - (j'_1 + 1 + (b-b')) \le j'' + r_1$; and 
  \item[S3] $0\le d_2 = j_2 - j'_2  \le w'' - r_1$. 
\end{itemize}

Again, if $f'$ and $f''$ are strongly $(b,j_1,j_2)$-compatible, we want their \emph{combined value}, $v(f',f'')$, to be the amount of flow that can be sent to $V(B)\setminus\{B_{x_0},\cdots,B_{x_{\ell'+1}},(B)\}$. Because $w''-r_1$ units of flow are necessarily sent to $(B)$, we know that this is either $q-(\ell'+1)$, or $w'+j''+r_1$ minus the flow used to satisfy the lack. This gives us that: $v(f',f'') = \min\{q - \ell' -1,w'+j'' + r_1 + r_2 - (b-b') - d_1\}$, where $r_2 = \min\{w''-r_1-d_2, d_G(x_{\ell'+1})-k+1-r_1\}$ (minimum between the amount of flow remaining in $V(B^{\ell'+1})\setminus\{(B^{\ell'+1})\}$ and the amount of flow that we can still push through $(B^{\ell'+1},B)$). If $f',f''$ are either weakly or strongly $(b,j_1,j_2)$-compatible, we say that they are \emph{$(b,j_1,j_2)$-compatible} (or just ``compatible'' if there is no ambiguity). Finally, we prove that these definitions completely describe our solution set.

\begin{lemma}
There exists a flow $f\in \mathcal{F}_{\ell'+1}$ that realizes $(b,j_1,j_2)$ of value $w+b+j_1+j_2$ if and only if there are flows $f'\in \mathcal{F}_{\ell'}$ and $f''\in \mathcal{F}(B^{\ell'+1})$ such that $f',f''$ are $(b,j_1,j_2)$-compatible and $v(f',f'')=w$.
\end{lemma}
\begin{proof}
$\Leftarrow$: Consider entries $e'=(b',j'_1,j'_2)$ in $P_{\ell'}$ and $e''=(b'',j'')$ in $S_{B^{\ell'+1}}$ such that $f',f''$ realize $e',e''$, respectively. Let $f'$ have value $b'+j'_1+j'_2+w'$ and $f''$ have value $b''+j''+w''$. First, suppose they are weakly compatible. Observe that the flow paths in $f'$ with extremity in $x_{\ell'+1},\cdots, x_q$ can actually end in any subset of these vertices, the same being valid for the flow paths in $f''$ with extremity in $x_1\cdots,x_{\ell'}$. We construct a flow in $\mathcal{F}_{\ell'+1}$ that realizes $(b,j_1,j_2)$ of value $v(f',f'')$ as follows. 
\begin{enumerate}
  \item Push $d_1$ units of flow from $V(B^{\ell'+1}) \setminus \{B^{\ell'+1}_{v_{\ell'+1}}\}$ into $(B^{\ell'+1}, B)$, then into any subset of $d_1$ unsaturated nodes in $\{B_{x_1},\cdots,B_{x_{\ell'}}\}$. This is possible by (W1), which implies $d_1+j'_1\le \ell'$, and by (W3), which ensures that there is a sufficient amount of flow to be pushed.
  \item Push $d_2$ units of flow from $V(B^{\ell'+1}) \setminus \{B^{\ell'+1}_{x_{\ell'+1}}\}$ into $(B^{\ell'+1}, B)$, then into $(B)$. This is possible by (W2) and (W3).
  \item If $b' = b$, then we know that $B_{x_0}$ is already satisfied by $f'$. Otherwise, push $b$ units of flow from $V(B^{\ell'+1}) \setminus \{B^{\ell'+1}_{x_{\ell'+1}}\}$ into $(B^{\ell'+1}, B)$, then into $B_{x_0}$. This is possible by (W3).
  \item Push flow $f''(B^{\ell'+1}_{v_{\ell'+1}})$ directly into $B_{x_{\ell'+1}}$, if there is any. If $j_1 = \ell'+1$ and $b''=0$, in which case $d'_1=1$, we push 1 unit of flow from $\{B_{x_{\ell'+1}},\cdots,B_{x_q}\}$ into $B_{x_{\ell'+1}}$. This is possible by (W3).
  \item Push the remaining flow into $\{B_{x_{\ell'+2}},\cdots,B_{x_q}\}$, decreasing the amount if needed, i.e., if there is more flow than nodes. This can be done by Proposition \ref{prop:1lessflow}. By the previous steps, we know that the remaining amount is $v(f',f'')$. 
\end{enumerate}

Now, suppose that $f',f''$ are strongly compatible. This implies $b'' = 0$. We construct a flow in $\mathcal{F}_{\ell'+1}$ where $\langle x_{\ell'+1}\rangle$ is a source as explained below. First, note that $x_{\ell'+1}$ is $k$-dense, by $S1$.

\begin{enumerate}
  \item Create a new source $\langle x_{\ell'+1}\rangle$ and push the new unit of flow into $B^{\ell'+1}_{x_{\ell'+1}}$, then into $B_{x_{\ell'+1}}$. This is possible since $b'' = 0$.
  \item Push $d_1+(b-b')$ units of flow from $(B^{\ell'+1})$, and $r_1$ other nodes in $V(B^{\ell'+1})$, into $B_{x_0},B_{x_1},\cdots,B_{x_{\ell'}}$; this is possible by (S2).
  \item Push $d_2$ units of flow from $V(B^{\ell'+1}) \setminus \{(B^{\ell'+1}), B^{\ell'+1}_{x_{\ell'+1}}\}$ into $(B)$; this is possible by (S3);
  \item Finally, push as much flow as possible from what remains in $V(B^{\ell'+1})$ into $\{B_{x_{\ell'+2}},\cdots,B_{x_q}\}$, decreasing the amount of flow if needed. It is possible to verify that this is equal to $v(f',f'')$.
\end{enumerate}

$\Rightarrow$: Consider a flow in $\mathcal{F}_{\ell'+1}(B)$ that realizes entry $(b,j_1,j_2)$, and let $f',f''$ be $f$ restricted to $\mathcal{F}_{\ell'},\mathcal{F}(B^{\ell'+1})$, respectively. In $f'$, let $b'=f'(B_{x_0})$, $j'_1=f'(\{B_{x_1},\cdots,B_{x_{\ell'}}\})$, $j'_2=f'((B))$, and $w'=f'(\{B_{x_{\ell'+1}},\cdots,B_{x_q}\})$. Clearly, $f'$ realizes entry $e'=(b',j'_1,j'_2)$ in $P_{\ell'}$. Now, let $j''=f''((B^{\ell'+1}))$, and $w''$ be the value of $f''$ minus $j''+f''(B^{\ell'+1}_{x_{\ell'+1}})$. Finally, let $b''$ be either the amount of flow received by $B^{\ell'+1}_{x_{\ell'+1}}$, if $\langle x_{\ell'+1}\rangle$ is not a source in $f$, or 0 otherwise. Clearly, $f''$ realizes entry $e'' = (b'',j'')$ in $S_{B^{\ell'+1}}$. We need to prove that $f',f''$ are compatible and $v(f',f'')=w$. Consider two cases: 

\begin{itemize}
  \item $\langle v_{\ell'+1}\rangle$ is not a source in $f$: clearly, $j_1\ge j'_1+f(B_{x_{\ell'+1}})\ge j'_1+b''$ and $j_2\ge j'_2$, i.e., (W1) and (W2) hold. Also, by definition, we know that the flow at $\{B_{x_0},\cdots,B_{x_{\ell'}},(B)\}$ that do not come from $f'$ must come from $f''$, that is, the first part of (W3) holds. Finally, if either $j_1\le \ell'$ or $b''=1$, then $w'\ge 0 \ge j_1 - \ell' - b''$ and (W3) follows. So, suppose that $j_1=\ell'+1$, and $b''=0$. Observe that, in this case, $B_{x_{\ell'+1}}$ receives flow from some node in $\mathcal{F}_{\ell'}(B)$, which implies $w'\ge 1 = j_1 - \ell' - b''$.
  
  \item $\langle v_{\ell'+1}\rangle$ is a source in $f$: Let $r'_1$ be the amount of flow sent from $(B^{\ell'+1},B)$ to $B_{x_0},B_{x_1},\cdots,B_{x_{\ell'}}$. By definition, we know that $f(\{B_{x_0},B_{x_1},\cdots,B_{x_{\ell'+1}}\})$ and $f((B))$ is the sum of $f'$ and $f''$ on these nodes; therefore, $j_1 - (j'_1 + f''(B_{x_{\ell'+1}}) + (b-b'))\ge 0$ and $j_2-j'_2\ge 0$. Also, the flow in $\{B_{x_0},\cdots,B_{x_{\ell'+1}}\}$ not coming from $f'$ and $\langle x_{\ell'+1}\rangle$, must come from $(B^{\ell'+1})$ and $(B^{\ell'+1},B)$, i.e., $j''+r_1\ge d_1$, and $(S2)$ holds. Similarly, the flow in $(B)$ not coming from $f'$ must come from $V(B^{\ell'+1})\setminus\{B^{\ell'+1_{x_{\ell'+1}}},(B^{\ell'+1})\}$, i.e., $w''-r_1\ge d_2$ and $(S3)$ holds. Finally, (S1) holds because of the capacity of node $(B^{\ell'+1},B)$ and it is not hard to verify that $v(f',f'')$ equals $w$.
\end{itemize}
\end{proof}

Consider $f',f''$ to be $(b,j_1,j_2)$ compatible. Note that only conditions (W3) and (S3) depend on the values $w',w''$, and that, if they hold for $w',w''$, they also hold for bigger values. Also, the larger these values are, the larger is the combined value of $f'$ and $f''$. Therefore, it indeed suffices to investigate tables $P_{\ell'}$ and $S_{B^{\ell'+1}}$ in order to compute $P_{\ell'+1}$. This gives us the complexity presented in the next corollary. We realize that this complexity can be refined, however here we are more concerned about the theoretical aspect of the problem.

\begin{corollary}
Given a claw-free block graph $G$ and an integer $k>\omega$, where $\omega = \omega(G)$, it can be decided whether $b(G)\ge k$ in time $O(\omega^4k^3n)$. 
\end{corollary}
\begin{proof}
Each table $P_{\ell'}$ has size $2\ell' (k-\lvert B\rvert) = O(\omega k)$, and each table $S_{B^{\ell'+1}}$ has size $2(k-\lvert B\rvert)=O(k)$. Also, deciding whether entries $(b',j'_1,j'_2)$ of $P_{\ell'}$ and $(b'',j'')$ of $S_{B^{\ell'+1}}$ are $(b,j_1,j_2)$-compatible takes time $d(v_{\ell'+1}) = O(\omega)$. Therefore, computing an entry $(b,j_1,j_2)$ of $P_{\ell'+1}$ takes time $O(\omega^2 k^2)$ and, because we need to compute $O(\omega k)$ entries, for each $\ell'$ such that $1\le \ell' \le \ell\le \omega$, it takes time $O(\omega^4k^3)$ to compute $S_B$. Finally, since there are $O(n)$ blocks in $G$, the theorem follows.
\end{proof}

%

\end{document}